\documentclass[sigconf]{acmart}

\usepackage{physics}
\usepackage{balance}
\usepackage{bbm}

\usepackage{array}
\newcolumntype{L}[1]{>{\raggedright\let\newline\\\arraybackslash\hspace{0pt}}m{#1}}
\newcolumntype{C}[1]{>{\centering\let\newline\\\arraybackslash\hspace{0pt}}m{#1}}
\newcolumntype{R}[1]{>{\raggedleft\let\newline\\\arraybackslash\hspace{0pt}}m{#1}}

\copyrightyear{2017} 
\acmYear{2017} 
\setcopyright{acmlicensed}
\acmConference{KDD '17}{August 13-17, 2017}{Halifax, NS, Canada}\acmPrice{15.00}\acmDOI{10.1145/3097983.3098095}
\acmISBN{978-1-4503-4887-4/17/08}

\DeclareMathOperator*{\argmax}{arg\,max}

\newcommand{\EE}{\mathbb{E}}
\newcommand{\RR}{\mathbb{R}}

\begin{document}

\title{Algorithmic decision making and the cost of fairness}

\author{Sam Corbett-Davies}
\affiliation{%
    \institution{Stanford University}
}
\email{scorbett@stanford.edu}

\author{Emma Pierson}
\affiliation{%
    \institution{Stanford University}
}
\email{emmap1@stanford.edu}

\author{Avi Feller}
\affiliation{%
    \institution{Univ. of California, Berkeley}
}
\email{afeller@berkeley.edu}

\author{Sharad Goel}
\affiliation{%
    \institution{Stanford University}
}
\email{scgoel@stanford.edu}

\author{Aziz Huq}
\affiliation{%
    \institution{University of Chicago}
}
\email{huq@uchicago.edu}

\begin{abstract}
Algorithms are now regularly used to decide whether defendants awaiting trial are too dangerous to be released back into the community. In some cases, black defendants are substantially more likely than white defendants to be incorrectly classified as high risk. 
To mitigate such disparities,
several techniques have recently been proposed to achieve
\emph{algorithmic fairness}. 
Here we reformulate algorithmic fairness as constrained optimization:
the objective is to maximize public safety 
while satisfying formal fairness constraints designed to reduce racial disparities.
We show that for several past definitions of fairness,
the optimal algorithms that result require 
detaining defendants above race-specific risk thresholds.
We further show that the optimal \emph{unconstrained} algorithm 
requires applying a single, uniform threshold to all defendants.
The unconstrained algorithm thus maximizes public safety while also satisfying one important understanding of equality: that all individuals are held to the same standard, irrespective of race.
Because the optimal constrained and unconstrained algorithms generally differ,
there is tension between improving public safety and satisfying prevailing notions of algorithmic fairness.
By examining data from Broward County, Florida,
we show that this trade-off can be large in practice.
We focus on algorithms for pretrial release decisions, but the principles we discuss apply to other domains, and also to human decision makers carrying out structured decision rules.
\end{abstract}

%
%

%

\maketitle
\renewcommand{\shortauthors}{Corbett-Davies et al.}

\section{Introduction}

Judges nationwide use algorithms to help 
decide whether defendants should be detained or released while awaiting trial~\cite{monahan2016,christin2015}.
One such algorithm, called COMPAS, assigns defendants risk scores between 1 and 10 that indicate how likely they are to commit a violent crime based on more than 100 factors, including age, sex and criminal history. 
For example, defendants with scores of 7 reoffend at twice the rate as those with scores of 3. 
Accordingly, defendants classified as high risk are much more likely to be detained while awaiting trial than those classified as low risk.

These algorithms do not explicitly use race as an input. Nevertheless, an analysis of defendants in Broward County, Florida \cite{angwin2016} revealed that black defendants are substantially more likely to be classified as high risk. Further, among defendants who ultimately did not reoffend, blacks were more than twice as likely as whites to be labeled as risky. Even though these defendants did not go on to commit a crime, being classified as high risk meant they were subjected to harsher treatment by the courts.
To reduce racial disparities of this kind, several authors recently have  proposed a variety of \emph{fair} decision algorithms~\cite{feldman2015,kamiran2013,hardt2016,joseph2016fairness,jabbari2016fair,hajian2011discrimination}.\footnote{We consider racial disparities because they have been at the center of many recent debates in criminal justice, but the same logic applies across a range of possible attributes, including gender.}

Here we reformulate algorithmic fairness as constrained optimization:
the objective is to maximize public safety 
while satisfying formal fairness constraints.
We show that for several past definitions of fairness,
the optimal algorithms that result require applying multiple, 
race-specific thresholds to individuals' risk scores.
One might, for example, detain white defendants who score above 4, but detain black defendants only if they score above 6.
We further show that the optimal \emph{unconstrained} algorithm requires applying a single, uniform threshold to all defendants.
This safety-maximizing rule thus satisfies one important understanding of equality: 
that all individuals are held to the same standard, irrespective of race.
Since the optimal constrained and unconstrained algorithms in general differ,
there is tension between reducing racial disparities 
and improving public safety.
By examining data from Broward County,
we demonstrate that this tension is more than theoretical.
Adhering to past fairness definitions can substantially decrease public safety; 
conversely, optimizing for public safety alone can produce stark racial disparities.

We focus here on the problem of designing algorithms for pretrial release decisions, but the principles we discuss apply to other domains, and also to human decision makers carrying out structured decision rules. We emphasize at the outset that algorithmic decision making does not preclude additional, or alternative, policy interventions. For example, one might provide released defendants with robust social services aimed at reducing recidivism, or conclude that it is more effective and equitable to replace pretrial detention with non-custodial supervision.
Moreover, regardless of the algorithm used, human discretion may be warranted in individual cases.

\section{Background}

\subsection{Defining algorithmic fairness}
\label{sec:def-fairness}

Existing approaches to algorithmic fairness typically
proceed in two steps. 
First, a formal criterion of fairness is defined;
then, a decision rule is developed to 
satisfy that measure, either exactly or approximately.
To formally define past fairness measures, we 
introduce a general notion of (randomized) decision rules.
Suppose we have a vector $x_i \in \RR^p$
that we interpret as the visible attributes of individual $i$.
For example, $x$ might represent a defendant's age, gender, race, and criminal history.
We consider binary decisions (e.g., $a_0=\text{release}$ and $a_1=\text{detain}$),
and define a \emph{decision algorithm}, or a \emph{decision rule}, to
be any function $d$ that specifies which action is taken for each individual.
To allow for probabilistic decisions, we require only that $d(x) \in [0,1]$.

\begin{definition}[Decision rule]
A decision rule is any measurable function $d: \RR^p \mapsto [0,1]$,
where we interpret $d(x)$ as the probability that action 
$a_1$ is taken for an individual with visible attributes $x$.
\end{definition}

Before defining algorithmic fairness, we need three additional concepts.
First, we define the \emph{group membership} of each individual to take a value from the set $\{g_1, \dots, g_k\}$. In most cases, we imagine these groups indicate an individual's race, but they might also represent gender or other protected attributes.
We assume an individual's racial  group can be inferred from their vector of observable attributes $x_i$, and so denote $i$'s group by $g(x_i)$.
For example, if we encode race as a coordinate in the vector $x$, 
then $g$ is simply a projection onto this coordinate.
Second, for each individual, we suppose there is a quantity
$y$ that specifies the benefit of taking action $a_1$ relative to action $a_0$.
For simplicity, we assume $y$ is binary and normalized to take values 0 and 1, but many of our
results can be extended to the more general case.
For example, in the pretrial setting, it is beneficial to
detain a defendant who would have committed a violent crime if released.
Thus, we might have $y_i = 1$ for those defendants who would have committed a violent crime if released, 
and $y_i=0$ otherwise.
Importantly, $y$ is not known exactly to the decision maker, who at the time of the decision 
has access only to information encoded in the visible features $x$.
Finally, we define random variables 
$X$ and $Y$ that take on values $X = x$ and $Y = y$ for an individual drawn randomly from the population of interest (e.g., the population of defendants for whom pretrial decisions must be made).

With this setup, we now describe three popular definitions of algorithmic fairness.

\begin{enumerate}
\item \emph{Statistical parity} means that an equal proportion of 
defendants are detained in each race group~\cite{kamishima2012,feldman2015,zemel2013,fish16}.
For example, white and black defendants are detained at equal rates. 
Formally, statistical parity means,
\begin{equation}
\EE[ d(X) \mid g(X) ] = \EE[ d(X) ].
\end{equation}

\item \emph{Conditional statistical parity} means that controlling for a limited set of ``legitimate'' risk factors, an equal proportion of defendants are detained within each race group~\cite{kamiran2013,dwork2012}.\footnote{%
Conditional statistical parity is closely related to the idea of \emph{fairness through blindness},
in which one attempts to create fair algorithms by prohibiting use of protected attributes, such as race. 
However, as frequently noted, it is difficult to restrict to ``legitimate'' features that do not at least partially correlate 
with race and other protected attributes, and so one
cannot be completely ``blind'' to the sensitive information~\cite{dwork2012}.
Moreover, unlike the other definitions of fairness,
this one does not necessarily reduce racial disparities. 
Conditional statistical parity mitigates these limitations
of the blindness approach while preserving its intuitive appeal.
}
For example, among defendants who have the same number of prior convictions, black and white defendants are detained at equal rates.  
Suppose $\ell : \RR^p \mapsto \RR^m$ is a projection of $x$ to factors considered legitimate. Then conditional statistical parity means,
\begin{equation}
\EE[ d(X) \mid \ell(X), \, g(X)] = \EE[ d(X) \mid \ell(X)].
\end{equation}

\item \emph{Predictive equality} means that the accuracy of decisions is equal across race groups, as measured by false positive rate (FPR)~\cite{hardt2016,zafar2016fairness,kleinberg2016inherent}. This condition means that among defendants who would not have gone on to commit a violent crime if released, detention rates are equal across race groups. Formally, predictive equality means,
\begin{equation}
\EE[ d(X) \mid Y=0, \, g(X)] = \EE[ d(X) \mid Y=0].
\end{equation}
As noted above, a major criticism of COMPAS is that the rate of false positives is higher among blacks than whites~\cite{angwin2016}.
\end{enumerate}

\subsection{Related work}

The literature on designing fair algorithms is extensive and interdisciplinary.
\citet{romei2014multidisciplinary} and \citet{vzliobaitemeasuring} survey various measures of fairness in decision making. 
Here we focus on algorithmic decision making in the criminal justice system, and briefly discuss several interrelated strands of past empirical and theoretical work.

Statistical risk assessment has been used in criminal justice for nearly one hundred years, 
dating back to parole decisions in the 1920s.
Several empirical studies have measured 
the effects of adopting such decision aids.
In a randomized controlled trial,
the Philadelphia Adult Probation and Parole
Department evaluated the effectiveness of a risk
assessment tool developed by Berk et al.~\cite{berk2009},
and found the tool 
reduced the burden on parolees without significantly
increasing rates of re-offense~\cite{ahlman2009}.
In a study by Danner et al.~\cite{danner2015}, pretrial services agencies in Virginia were randomly chosen to adopt supervision guidelines based on a risk assessment tool. 
Defendants processed by the chosen agencies were nearly twice as likely to be released, and these released defendants were on average less risky than 
those released by agencies not using the tool.
We note that despite
such aggregate benefits, 
some have argued that statistical tools do not provide 
sufficiently precise estimates of individual recidivism risk 
to ethically justify their use~\cite{starr2014evidence}.\footnote{%
Eric Holder, former Attorney General of the United States, has
been similarly critical of risk assessment tools, arguing
that ``[e]qual justice can only mean individualized
justice, with charges, convictions, and sentences
befitting the conduct of each defendant
and the particular crime he or she commits''~\cite{holder2014}.
}

Several authors have developed algorithms that guarantee formal definitions of fairness are satisfied. To ensure statistical parity,
\citet{feldman2015} propose ``repairing''
attributes or risk scores by converting them to within-group percentiles.
For example, a black defendant riskier than 90\% of black defendants would receive the same transformed score as a white defendant riskier than 90\% of white defendants.
A single decision threshold applied to the transformed scores would then result in equal detention rates across groups.
\citet{kamiran2013} propose a similar method (called ``local massaging'') to achieve conditional statistical parity. 
Given a set of decisions, they 
stratify the population by ``legitimate'' factors (such as number of prior convictions), and then alter decisions within each stratum so that:
(1) the overall proportion of people detained within each stratum remains unchanged;
and (2) the detention rates in the stratum are equal across race groups.\footnote{In their context, they consider human decisions, rather than algorithmic ones, but the same procedure can be applied to any rule.}
Finally, \citet{hardt2016} propose a method for constructing
randomized decision rules that ensure true positive and false positive rates
are equal across race groups, a criterion of fairness that they call \emph{equalized odds}; they further study the case in which only true positive rates must be equal, which they call \emph{equal opportunity}. 

The definitions of algorithmic fairness
discussed above assess the
fairness of \emph{decisions}; in contrast, 
some authors consider the fairness of \emph{risk scores}, 
like those produced by COMPAS.
The dominant fairness criterion in this case is \emph{calibration}.\footnote{Calibration is sometimes called \emph{predictive parity}; we use ``calibration'' here to distinguish it from predictive \emph{equality}, meaning equal false positive rates.}
Calibration means that among defendants with a given risk score,
the proportion who reoffend is the same across race groups.
Formally, given risk scores $s(X)$, calibration means,
\begin{equation}
\label{eq:calibration}
\Pr(Y = 1 \mid s(X), \, g(X)) = \Pr(Y = 1 \mid s(X)).
\end{equation}
Several researchers have pointed out that many notions of fairness are in conflict; \citet{berk2017} survey various fairness measures and their incompatibilities.
Most importantly, \citet{kleinberg2016inherent} prove that except in degenerate cases,
no algorithm can simultaneously satisfy the following three properties: 
(1) calibration;
(2) balance for the negative class, meaning that among defendants who would not commit a crime if released, average risk score is equal across race group;
and (3) balance for the positive class, meaning that among defendants who would commit a crime if released, average risk score is equal across race group.
\citet{chouldechova2016fair} similarly considers the tension between calibration and alternative definitions of fairness.

\section{Optimal decision rules}
Policymakers wishing to satisfy a particular definition of fairness are necessarily restricted in the set of decision rules that they can apply. 
In general, however, multiple rules satisfy any given fairness criterion, 
and so one must still decide which rule 
to adopt from among those satisfying the constraint.
In making this choice, we assume policymakers seek to maximize 
a specific notion of utility, which we detail below.

In the pretrial setting, one must balance two factors: 
the benefit of preventing violent crime committed by released defendants on the one hand,
and the social and economic costs of detention on the other.\footnote{Some jurisdictions consider flight risk, but safety is typically the dominant concern.} 
To capture these costs and benefits, we define the \emph{immediate utility} of a decision rule as follows.

\begin{definition}[Immediate utility]
For $c$ a constant such that $0 < c < 1$,
the immediate utility of a decision rule $d$ is
\begin{align}
u(d, c) & = \EE \left [ Yd(X) - cd(X) \right ] \nonumber \\
& = \EE \left [ Yd(X) \right ] - c\EE \left [ d(X) \right ]. \label{eq:utility}
\end{align}
\end{definition}

The first term in Eq.~\eqref{eq:utility}
is the expected benefit of the decision rule, and the second term its costs.\footnote{%
We could equivalently define immediate utility in terms of the relative costs 
of false positives and false negatives, but we believe our formulation
 better reflects the concrete trade-offs policymakers face.
}
For pretrial decisions, the first term is proportional to the expected number of violent crimes prevented under $d$,
and the second term is proportional to the expected number of people detained.
The constant $c$ is the cost of detention in units of crime prevented.
We call this \emph{immediate} utility to clarify that it reflects only the proximate 
costs and benefits of decisions. It does not, for example, consider the long-term,
systemic effects of a decision rule.

 We can rewrite immediate utility as
 \begin{align}
 u(d, c) & = \EE \left [ \EE \left [ Yd(X) - cd(X) \mid X \right ] \right ] \nonumber \\
 & = \EE \left [ p_{Y|X} d(X) - cd(X) \right ] \nonumber \\
 & = \EE \left [ d(X) (p_{Y|X} - c) \right ] \label{eq:utility2}
 \end{align}
 where $p_{Y|X} = \Pr(Y = 1 \mid X)$.
This latter expression shows that it is beneficial to detain an individual precisely when $p_{Y|X} > c$, and is a convenient reformulation for our derivations below.

Our definition of immediate utility implicitly encodes two important assumptions.
First, since $Y$ is binary, all violent crime is assumed to be equally costly.
Second, the cost of detaining every individual is assumed to be $c$, 
without regard to personal characteristics. 
Both of these restrictions can be relaxed without
significantly affecting our formal results.
In practice, however, it is often difficult to approximate 
individualized costs and benefits of detention, and 
so we proceed with this framing of the problem.

Among the rules that satisfy a chosen fairness criterion, 
we assume policymakers would prefer the one that maximizes immediate utility.
For example, if policymakers wish to ensure statistical parity, they might first consider all decision rules that guarantee statistical parity is satisfied, 
and then adopt the utility-maximizing rule among this subset.

For the three fairness definitions we consider (statistical parity, conditional statistical parity, and predictive equality) 
we show next 
that the optimal algorithms that result are simple, deterministic threshold rules based on $p_{Y\mid X}$.
For statistical parity and predictive equality, 
the optimal algorithms detain defendants when $p_{Y\mid X}$ exceeds 
a group-specific threshold.
For example, black defendants might be detained if $p_{Y\mid X} \geq 0.2$,
and white defendants detained if $p_{Y\mid X} \geq 0.1$.
The exact thresholds for statistical parity differ from those for predictive equality.
For conditional statistical parity, the thresholds in the optimal decision rule depend on both
group membership and the ``legitimate'' factors $\ell(X)$. 
Finally, we show that the unconstrained utility-maximizing algorithm 
applies a single, uniform threshold to all individuals, irrespective of group membership.
Importantly, since the optimal constrained algorithms differ from the optimal unconstrained algorithm, fairness has a cost.

To prove these results, we require one more technical criterion: 
that the distribution of $p_{Y|X}$ has a strictly positive density on $[0,1]$.
Intuitively, $p_{Y|X}$ is the risk score for a randomly selected individual with visible attributes $X$. Having a density means that the distribution of $p_{Y|X}$ does not have any point masses: for example, the probability that $p_{Y|X}$ \emph{exactly} equals 0.1 is zero. Positivity means that in any sub-interval, there is non-zero (though possibly small) probability an individual has risk score in that interval.
From an applied perspective, this is a relatively weak condition, since 
starting from any risk distribution 
we can achieve this property by smoothing the distribution 
by an arbitrarily small amount.
But the criterion serves two important technical purposes.
First, with this assumption, there are always deterministic decision rules that satisfy each fairness definition; and
second, it implies that the optimal decision rules are unique.
We now state our main theoretical result.

\begin{theorem}
\label{thm:mother}
Suppose $\mathcal{D}(p_{Y|X})$ has positive density on $[0,1]$. 
The optimal decision rules $d^*$ that maximize $u(d,c)$ under various fairness conditions have the following form, 
and are unique up to a set of probability zero.
\begin{enumerate}
\item The unconstrained optimum is
\begin{equation*}
d^*(X) = 
\begin{cases}
1 & p_{Y|X} \geq c \\
0 & \textnormal{otherwise}
\end{cases}
\end{equation*}
\item Among rules satisfying statistical parity, the optimum is
\begin{equation*}
d^*(X) = 
\begin{cases}
1 & p_{Y|X} \geq t_{g(X)} \\
0 & \textnormal{otherwise}
\end{cases}
\end{equation*}
where $t_{g(X)} \in [0,1]$ are constants that depend only on group membership. The optimal rule satisfying predictive equality takes the same form, though the values of the group-specific thresholds are different.
\item Additionally suppose $\mathcal{D}(p_{Y|X}\mid \ell(X)=l)$ has positive density on $[0,1]$. 
Among rules satisfying conditional statistical parity, the optimum is
\begin{equation*}
d^*(X) = 
\begin{cases}
1 & p_{Y|X} \geq t_{g(X), \ell(X)} \\
0 & \textnormal{otherwise}
\end{cases}
\end{equation*}
where $t_{g(X),\ell(X)} \in [0,1]$ are constants that depend on group membership and ``legitimate'' attributes.
\end{enumerate}
\end{theorem}

Before presenting the formal proof of Theorem~\ref{thm:mother}, we sketch out the argument. From Eq.~\eqref{eq:utility2}, it follows immediately that (unconstrained) utility is maximized for a rule that deterministically detains defendants if and only if $p_{Y|X} \geq c$.
The optimal rule satisfying statistical parity necessarily detains the same proportion $p^*$ of defendants in each group; 
it is thus clear that utility is maximized by setting the thresholds so that the riskiest proportion $p^*$ of defendants is detained in each group.
Similar logic establishes the result for conditional statistical parity. 
(In both cases, our assumption on the distribution of the risk scores ensures these thresholds exist.)
The predictive equality constraint 
is the most complicated to analyze.
Starting from any non-threshold rule $d$ satisfying predictive equality, 
we show that one can derive a rule $d'$ satisfying predictive equality
such that $u(d', c) > u(d,c)$; this in turn implies a threshold rule is optimal.
We construct $d'$ in three steps.
First, we show that under the original rule $d$ there must exist some low-risk defendants that are detained while some relatively high-risk defendants
are released.
Next, we show that if $d'$ has the same false positive rate as $d$,
then $u(d', c) > u(d,c)$ if and only if more defendants are detained under $d'$.
This is because having equal false positive rates means that $d$ and $d'$ 
detain the same number of people who would not have committed a violent crime if released;
under this restriction, detaining more people means detaining more people who would have committed
a violent crime, which improves utility.
Finally, we show that one can preserve false positive rates by releasing the low-risk individuals
and detaining an even greater number of the high-risk individuals;
this last statement follows because releasing low-risk individuals decreases the false positive rate
faster than detaining high-risk individuals increases it.

\begin{proof}

As described above, it is clear that threshold rules are optimal absent fairness constraints, and also in the case of statistical parity and conditional statistical parity.
We now establish the result for predictive equality; we then prove the uniqueness of these rules.

Suppose $d$ is a decision rule satisfying equal false positive rates
and which is not equivalent to a multiple-threshold rule.
We shall construct a new decision rule $d'$
satisfying equal false positive rates, and such that
$u(d',c) > u(d,c)$.
Since this construction shows any non-multiple-threshold rule can be improved, the optimal rule must be a multiple-threshold rule.

Because $d$ is not equivalent to a multiple-threshold rule, there exist relatively low-risk individuals that are detained and relatively high-risk 
individuals that are released.
To see this, define $t_a$ to be the threshold that detains the same proportion of group $a$ as $d$ does:
\begin{equation*}
\EE \left [d(X) \mid g(X)=a\right] = \EE \left [\mathbbm{1}\{p_{Y|X}\geq t_a\} \mid g(X)=a\right].
\end{equation*}
Such thresholds exist by our assumption on the distribution of $p_{Y|X}$.
Since $d$ is not equivalent to a multiple-threshold rule, there must be a group $a^*$ for which, in expectation, some defendants below $t_{a^*}$ will be detained and an equal proportion of defendants above $t_{a^*}$ released. Let $2\beta$ equal the proportion of defendants ``misclassified" (with respect to $t_{a^*}$) in this way:
\begin{align*}
\beta &= \EE \left[\mathbbm{1}\{p_{Y|X} \geq t_{a^*}\}(1-d(X)) \mid g(X)=a^*\right]\\
&= \EE \left[\mathbbm{1}\{p_{Y|X}< t_{a^*}\}d(X) \mid g(X)=a^*\right]\\
&>0,
\end{align*}
where we note that $\Pr( p_{Y|X} = t_{a^*}) = 0$.

For $0 \leq t_1 \leq t_2 \leq 1$, define the rule
\begin{equation*}
d'_{t_1,t_2}(X) = \begin{cases}
1 & p_{Y|X}\geq t_2,\, g(X)=a^* \\
0 & p_{Y|X}< t_1,\, g(X)=a^* \\
d(X) & \textnormal{otherwise}
\end{cases}.
\end{equation*}
This rule detains
\begin{align*}
\beta_2(t_1,t_2)=\EE \left[\mathbbm{1}\{p_{Y|X}\geq t_2\}(1-d(X)) \mid g(X)=a^*\right]
\end{align*}
defendants above the threshold who were released under $d$. Further,
\begin{align*}
\gamma_2(t_1,t_2) &=\EE \left[ \mathbbm{1}\{p_{Y|X}\geq t_2\}(1-d(X))(1-p_{Y|X}) \mid g(X)=a^*\right]\\
&\leq(1-t_2)\beta_2(t_1,t_2)
\end{align*}
defendants are newly detained and ``innocent'' (i.e.,  would not have gone on to commit a violent crime). Similarly, $d'_{t_1,t_2}$ releases
\begin{align*}
\beta_1(t_1,t_2)=\EE \left[\mathbbm{1}\{p_{Y|X}< t_1\}d(X) \mid g(X)=a^*\right]
\end{align*}
defendants below the threshold that were detained under $d$, resulting in
\begin{align*}
\gamma_1(t_1,t_2) &=\EE \left[ \mathbbm{1}\{p_{Y|X}< t_1\}d(X)(1-p_{Y|X}) \mid g(X)=a^*\right]\\
&\geq(1-t_1)\beta_1(t_1,t_2)
\end{align*}
fewer innocent detainees.

Now choose $t_1< t_{a^*} < t_2$ such that $\beta_1(t_1, t_2)=\beta_2(t_1,t_2)=\beta/2.$
Such thresholds exist because: $\beta_1(t_{a^*},t_{a^*})=\beta_2(t_{a^*},t_{a^*})=\beta$, $\beta_1(0,\cdot) = \beta_2(\cdot,1) = 0$, 
and the functions $\beta_i$ are continuous in each coordinate. 
Then, $\gamma_1(t_1,t_2) \geq (1-t_1)\beta/2$ and $\gamma_2(t_1,t_2) \leq (1-t_2)\beta/2$, so $\gamma_1(t_1,t_2) > \gamma_2(t_1,t_2)$. 
This inequality implies that $d'_{t_1,t_2}$ releases more innocent low-risk people than it detains innocent high-risk people (compared to $d$). 

To equalize false positive rates between $d$ and $d'$ we must equalize $\gamma_1$ and $\gamma_2$, and so we need to decrease $t_1$ in order to release fewer low-risk people. 
Note that $\gamma_1$ is continuous in each coordinate, 
$\gamma_1(0, \cdot)=0$, and 
$\gamma_2$ depends only on its second coordinate.
There thus exists $t_1' \in [0,t_1)$ such that $\gamma_1(t_1',t_2) = \gamma_2(t_1, t_2) = \gamma_2(t_1',t_2)$.
Further, since $t_1' < t_1$,
$\beta_1(t_1',t_2) < \beta_1(t_1,t_2) = \beta_2(t_1,t_2)$.
Consequently, $d'_{t_1',t_2}$ has the same false positive rate as $d$ but detains more people.

Finally, since false positive rates are equal, detaining extra people means detaining more people who go on to commit a violent crime. As a result 
$d'_{t_1',t_2}$ has strictly higher immediate utility than $d$:
\begin{align*}
u(d'_{t_1',t_2},c)-u(d,c)
& =\EE \left[d'_{t_1',t_2}(X)(p_{Y|X} - c)\right] -  \EE \left[d(X)(p_{Y|X} - c)\right] \\
&= \EE \left[d'_{t_1',t_2}(X)(1 - c)\right] - \EE \left[d'_{t_1',t_2}(X)(1 - p_{Y|X})\right]  \\
&\qquad - \EE \left[d(X)(1 - c)\right] + \EE \left[d(X)(1 - p_{Y|X})\right]\\
&= (1-c)\left( \EE \left[d'_{t_1',t_2}(X)\right]-\EE \left[d(X)\right]\right)\\
&= (1-c)\left[\beta_2(t_1',t_2)-\beta_1(t_1',t_2)\right]\\
&>0.
\end{align*}
The second-to-last equality follows from the fact that
$d'_{t_1',t_2}$ and $d$ have equal false positive rates, which
in turn implies that
\begin{align*}
\EE \left[d'_{t_1',t_2}(X)(1 - p_{Y|X} )\right]  =  \EE \left[d(X)(1 - p_{Y|X} )\right].
\end{align*}
Thus, starting from an arbitrary non-threshold rule satisfying predictive equality, 
we have constructed a threshold rule with strictly higher utility that also satisfies predictive equality; as a consequence, threshold rules are optimal.

We now establish uniqueness of the optimal rules for each fairness constraint.
Optimality for the unconstrained algorithm is clear, and so we consider only the constrained rules, starting with statistical parity.
Denote by $d_\alpha$ the rule that detains the riskiest proportion $\alpha$ of individuals in each group; this rule is the unique 
optimum among those with detention rate $\alpha$ satisfying statistical parity. Define 
\begin{align*}
f(\alpha) &= u(d_\alpha, c) \\
&= \EE \left[d_\alpha(X)p_{Y|X}\right] - c \alpha.
\end{align*}
The first term of $f(\alpha)$ is strictly concave, because 
$d_\alpha$ detains progressively less risky people as $\alpha$ increases. The second term of $f(\alpha)$ is linear. Consequently, $f(\alpha)$ is strictly concave and has a unique maximizer. 
A similar argument shows uniqueness of the optimal rule for conditional statistical parity.

To establish uniqueness in the case of predictive equality, 
we first restrict to the set of threshold rules, 
since we showed above that non-threshold rules are suboptimal.
Let $d_{\sigma}$ be the unique, optimal threshold rule having false positive rate $\sigma$ in each group.
Now let $g(\sigma)$ be the detention rate under $d_{\sigma}$.
Since $g$ is strictly increasing, there is 
a unique, optimal threshold rule $d'_{\alpha}$ that satisfies predictive equality and detains a proportion $\alpha$ of defendants:
namely, $d'_{\alpha} = d_{g^{-1}(\alpha)}$.
Uniqueness now follows by the same argument we gave for statistical parity.
\end{proof}

Theorem~\ref{thm:mother} shows that threshold rules maximize immediate utility when the three fairness criteria we consider hold \emph{exactly}.
Threshold rules are also optimal if we only require the constraints hold \emph{approximately}.
For example, a threshold rule maximizes immediate utility under the requirement that false positive rates differ by at most a constant $\delta$ across groups.
To see this, note that our construction in Theorem~\ref{thm:mother} preserves false positive rates. Thus, starting from a non-threshold rule that satisfies the (relaxed) constraint, one can construct a threshold rule that satisfies the constraint and strictly improves immediate utility, establishing the optimality of threshold rules.

Threshold rules have been proposed previously
to achieve the various fairness criteria we analyze~\cite{feldman2015,kamiran2013,hardt2016}.
We note two important distinctions between our work and past research.
First, the optimality of such algorithms has not been previously established, 
and indeed previously proposed decision rules are not always optimal.\footnote{%
Feldman et al.'s~\cite{feldman2015} algorithm for achieving statistical parity is optimal only if one ``repairs'' risk scores $p_{Y|X}$ rather than individual attributes.
Applying Kamiran et. al's local massaging algorithm~\cite{kamiran2013} for achieving conditional statistical parity yields a non-optimal multiple-threshold rule, even if one starts with the optimal single threshold rule.
\citet{hardt2016} hint at the optimality of their algorithm for achieving predictive equality---and in fact their algorithm is optimal---but they do not provide a proof.
}
Second, our results clarify the need for race-specific decision 
thresholds to achieve prevailing notions of algorithmic fairness.
We thus identify an inherent tension between satisfying common
fairness constraints and treating all individuals equally, irrespective of race.

Our definition of immediate utility does not put a hard cap on the number of people detained, 
but rather balances detention rates with public safety benefits via the constant $c$.
Proposition~\ref{prop:utility-equiv} below shows that one can equivalently view the optimization
problem as maximizing public safety while detaining a specified number of individuals.
As a consequence, the results in Theorem~\ref{thm:mother}---where immediate utility is maximized under a fairness constraint---also hold when public safety is optimized under constraints on both fairness and the proportion of defendants detained.
This reformulation is useful for our empirical analysis in Section~\ref{sec:empirical}.

\begin{proposition}
\label{prop:utility-equiv}
Suppose $D$ is the set of decision rules satisfying 
statistical parity, conditional statistical parity, predictive equality,
or the full set of all decision rules. 
There is a bijection $f$ on the interval $[0,1]$ such that
\begin{equation}
\label{eq:equiv}
\argmax_{d \in D} \EE \left [ Y d(X) - cd(X) \right ] 
= 
\argmax_{\substack{d \in D\\ \EE [d(X)] = f(c)}} 
\EE \left [ Y d(X) \right ] 
\end{equation}
where the equivalence of the maximizers in \eqref{eq:equiv} is defined up to a set of probability zero.
\end{proposition}

\begin{proof}
Let $f(c) = \EE \left[d^*(X)\right]$, where $d^*$ is the unique maximizer of $u(d, c)$ under the constraint $d \in D$.
For a fixed $c$, if a decision rule maximizes the right-hand side of \eqref{eq:equiv} then it is straightforward to see that it also maximizes the left-hand side. By uniqueness of the solution to the left-hand side, the solution to the right-hand side is also unique.
The equality in Eq.~\eqref{eq:equiv} thus holds for all $c$.

It remains to be shown that $f$ is a bijection. 
For fixed $c$ and $\alpha$, 
the proof of Theorem~\ref{thm:mother}
established that there is a unique, utility-maximizing 
threshold rule $d_{\alpha} \in D$ that detains a fraction $\alpha$ of individuals.
Let $g(\alpha) = u(d_{\alpha}, c)$.
Now,
\begin{align*}
g'(\alpha) &= \dv{\alpha} \EE \left[Yd_{\alpha}(X)-cd_{\alpha}(X)\right]\\
&= \dv{\alpha} \left(\EE \left[Yd_{\alpha}(X)\right]-c\alpha\right)
\end{align*}
and so $g(\alpha)$ is maximized at $\alpha^*$ such that
\begin{equation*}
\dv{\alpha} \EE \left[Yd_{\alpha}(X)\right]=c
\end{equation*}
In other words, the optimal detention rate $\alpha^*$ is such that the marginal person detained has probability $c$ of reoffending. Thus, as $c$ decreases, the optimal detention threshold decreases, and the proportion detained increases. Consequently, if $c_1 < c_2$ then $f(c_1) > f(c_2)$, 
and so $f$ is injective.
To show that $f$ is surjective, note that $f(0) = 1$ and $f(1) = 0$; the result now follows from continuity of $f$.
\end{proof}

\section{The cost of fairness}
\label{sec:empirical}

As shown above, the optimal algorithms under past notions of fairness differ from the unconstrained solution.\footnote{%
One can construct examples in which the group-specific thresholds coincide, 
leading to a single threshold, but it is unlikely for the thresholds to be \emph{exactly} equal in practice. 
We discuss this possibility further in Section~\ref{sec:incompatibility}.
}
Consequently, satisfying common definitions of fairness means one must in theory sacrifice some degree of public safety. 
We turn next to the question of how great this public safety loss might be in practice.

We use data from Broward County, Florida originally compiled by ProPublica~\cite{larson2016}. 
Following their analysis, we only consider black and white defendants who were assigned COMPAS risk scores within 30 days of their arrest, 
and were not arrested for an ordinary traffic crime.
We further restrict to only those defendants who spent at least two years (after their COMPAS evaluation) outside a correctional facility without being arrested for a violent crime, or were arrested for a violent crime within this two-year period.
Following standard practice, we use this two-year violent recidivism metric to approximate the benefit $y_i$ of detention: 
we set $y_i=1$ for those who reoffended, and $y_i=0$ for those who did not.
For the 3,377 defendants satisfying these criteria, 
the dataset includes race, age, sex, number of prior convictions, 
and COMPAS violent crime risk score (a discrete score between 1 and 10). 

The COMPAS scores may not be the most accurate estimates of risk, both because the scores are discretized and because they are not trained specifically for Broward County. 
Therefore, to estimate $p_{Y|X}$ we re-train a risk assessment model that predicts two-year violent recidivism using 
$L^1$-regularized logistic regression followed by Platt scaling~\cite{platt99}.
The model is based on all available features for each defendant, excluding race. 
Our risk scores achieve higher AUC on a held-out set of defendants than the COMPAS scores (0.75 vs. 0.73). 
We note that adding race to this model does not improve performance, as measured by AUC on the test set.

\begin{table}[t]
\centerline{
\begin{tabular}{p{2.5cm} C{2.7cm} C{2.4cm}}
Constraint & Percent of detainees that are low risk & Estimated increase in violent crime \\ 
  \hline
Statistical parity & 17\% & 9\% \\ 
Predictive equality & 14\% & 7\% \\ 
Cond. stat. parity & 10\% & 4\% \\ 
\end{tabular}
}
\vspace{4mm}
\caption{Based on the Broward County data, satisfying common fairness definitions results in detaining low-risk defendants while reducing public safety.
For each fairness constraint, we estimate the increase in violent crime committed by released defendants, relative to a rule that optimizes for public safety alone; 
and the proportion of detained defendants that are low risk 
(i.e., would be released if we again considered only public safety).
}
\label{table:cost-of-fairness}
\end{table}

We investigate the three past fairness definitions previously discussed: 
statistical parity, conditional statistical parity, and predictive equality. 
For each definition, we find the set of thresholds that produce a decision rule that: 
(1) satisfies the fairness definition;
(2) detains 30\% of defendants; 
and (3) maximizes expected public safety subject to (1) and (2). 
The proportion of defendants detained is chosen to match the fraction of defendants classified as medium or high risk by COMPAS (scoring 5 or greater).
Conditional statistical parity requires that one define the ``legitimate'' factors $\ell(X)$,
and this choice significantly impacts results. 
For example, if all variables are deemed legitimate, then this fairness condition imposes no constraint on the algorithm. 
In our application, we consider only a defendant's number of prior convictions to be legitimate;
to deal with sparsity in the data, we partition prior convictions into four bins: 0, 1--2, 3--4, and 5 or more. 

We estimate two quantities for each decision rule: 
the increase in violent crime committed by released defendants, relative to a rule that optimizes for public safety alone, ignoring formal fairness requirements; and the proportion of detained defendants that are low risk (i.e., would be released if we again considered only public safety).
We compute these numbers on 100 random train-test splits of the data. 
On each iteration, we train the risk score model and find the optimal thresholds using 70\% of the data,
and then calculate the two statistics on the remaining 30\%. 
Ties are broken randomly when they occur, and we report results averaged over all runs.

For each fairness constraint, Table~\ref{table:cost-of-fairness} shows that
violent recidivism increases while low risk defendants are detained. 
For example, when we enforce statistical parity, 17\% of detained defendants are relatively low risk. 
An equal number of high-risk defendants are thus released (because we hold fixed the number of individuals detained), 
leading to an estimated 9\% increase in violent recidivism among released defendants.
There are thus tangible costs to satisfying popular notions of algorithmic fairness.

\section{The cost of public safety}
\label{sec:incompatibility}

A decision rule constrained to satisfy statistical parity, conditional statistical parity, or predictive equality  reduces public safety. However, a single-threshold rule that maximizes public safety generally violates all of these fairness definitions. 
For example, in the Broward County data, optimally detaining 30\% of defendants with a single-threshold rule means that
40\% of black defendants are detained, compared to 18\% of white defendants, violating statistical parity. 
And among defendants who ultimately do not go on to commit a violent crime, 
14\% of whites are detained compared to 32\% of blacks, violating predictive equality.

\begin{figure}[t]
\centering
\includegraphics[width=\columnwidth, trim=0 1.2cm 0 0.5cm]{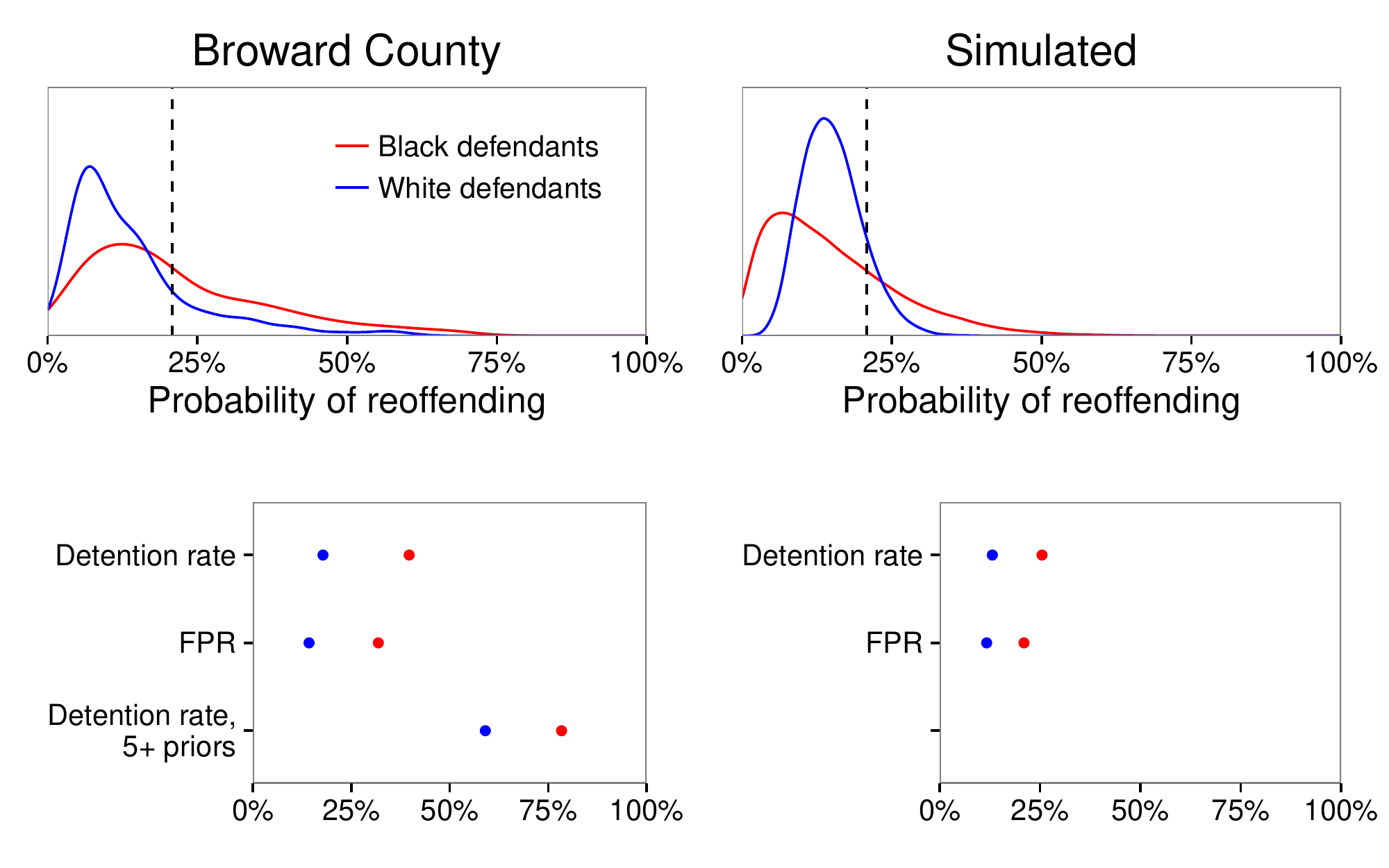}
\vspace{2mm}
\caption{
Top: distribution of risk scores for Broward County data (left), and simulated data drawn from two beta distributions with equal means (right). Bottom: using a single threshold which detains 30\% of defendants in Broward County violates statistical parity (as measured by detention rate), predictive equality (false positive rate), and conditional statistical parity (detention rate conditional on number of prior arrests). We omit the last measure for the simulated data since that would require making additional assumptions about the relationship of priors and risk in the hypothetical populations.
}
\label{fig:risk-scores}
\end{figure}

The reason for these disparities is that white and black defendants in Broward County have different distributions of risk, $p_{Y|X}$, as shown in Figure 1. In particular, a greater fraction of black defendants have relatively high risk scores, in part because black defendants are more likely to have prior arrests, which is a strong indicator of reoffending. Importantly, while an algorithm designer can choose different decision rules based on these risk scores, the algorithm cannot alter the risk scores themselves, which reflect underlying features of the population of Broward County.

Once a decision threshold is specified, these risk distributions determine the statistical properties of the decision rule,
including the group-specific detention and false positive rates.
In theory, it is possible that these distributions line up in a way that achieves statistical parity or predictive equality,
but in practice that is unlikely.
Consequently, any decision rule that guarantees these various fairness criteria are met will in practice deviate from the unconstrained optimum.

\citet{kleinberg2016inherent} establish the incompatibility of different fairness measures when the overall risk $\Pr(Y=1 \mid g(X) = g_i)$ differs between groups $g_i$. However, the tension we identify between maximizing public safety 
and satisfying various notions of algorithmic fairness
typically persists even if groups have the same overall risk. 
To demonstrate this phenomenon, Figure~\ref{fig:risk-scores} shows 
risk score distributions for two hypothetical populations with equal average risk. 
Even though their means are the same, the tail of the red distribution is heavier than the tail of the blue distribution,
resulting in higher detention and false positive rates in the red group.

That a single decision threshold can, and generally does, result in racial disparities is closely related to the 
notion of \emph{infra-marginality} in the econometric literature on taste-based discrimination~\cite{ayres2002outcome,simoiu2017,anwar2006,pierson2017}.
In that work, taste-based discrimination~\cite{becker1957}
is equated with applying decision thresholds that differ by race.
Their setting is human, not algorithmic, decision making, and so one cannot directly observe the thresholds being applied;
the goal is thus to infer the thresholds from observable statistics.
Though intuitively appealing, 
detention rates and false positive rates are poor proxies for the thresholds:
these infra-marginal statistics consider \emph{average} risk above the thresholds, and so can differ even if the thresholds are identical (as shown in Figure~\ref{fig:risk-scores}). 
In the algorithmic setting, past fairness measures notably focus on these infra-marginal statistics, even though the thresholds themselves are directly observable.

\section{Detecting discrimination}
\label{sec:calibration}

The algorithms we have thus far considered output a decision $d(x)$ for each individual.\pagebreak[0]
In practice, however, algorithms like COMPAS typically output a score $s(x)$
that is claimed to indicate a defendant's risk $p_{Y|X}$;\pagebreak[0]
decision makers then use these risk estimates to select an action (e.g., release or detain).

In some cases, neither the procedure nor the data used to generate these scores is disclosed, prompting worry that the scores are themselves discriminatory. 
To address this concern, researchers often examine whether
scores are calibrated~\cite{kleinberg2016inherent}, as defined by Eq.~\eqref{eq:calibration}.\footnote{%
Some researchers also check whether the AUC of scores is
similar across race groups~\cite{skeem2015risk}. 
The theoretical motivation for examining AUC is less clear, 
since the true risk distributions might have different AUCs,
a pattern that would be reproduced in scores that approximate these probabilities.
In practice, however, one might expect the true risk distributions to 
yield similar AUCs across race groups---and indeed this is the case for the Broward County data.
}
Since the true probabilities $p_{Y|X}$ are necessarily calibrated,
it is reasonable to expect risk estimates that approximate these probabilities to be calibrated as well. 
Figure~\ref{fig:calibration} shows that the COMPAS scores indeed satisfy this property.
For example, among defendants who scored a seven on the COMPAS scale, 
60\% of white defendants reoffended, which is nearly identical to the 61\% percent of black defendants who reoffended.

\begin{figure}[t]
\centering
\includegraphics[width=.8\columnwidth, trim = 0 1.1cm 0 0]{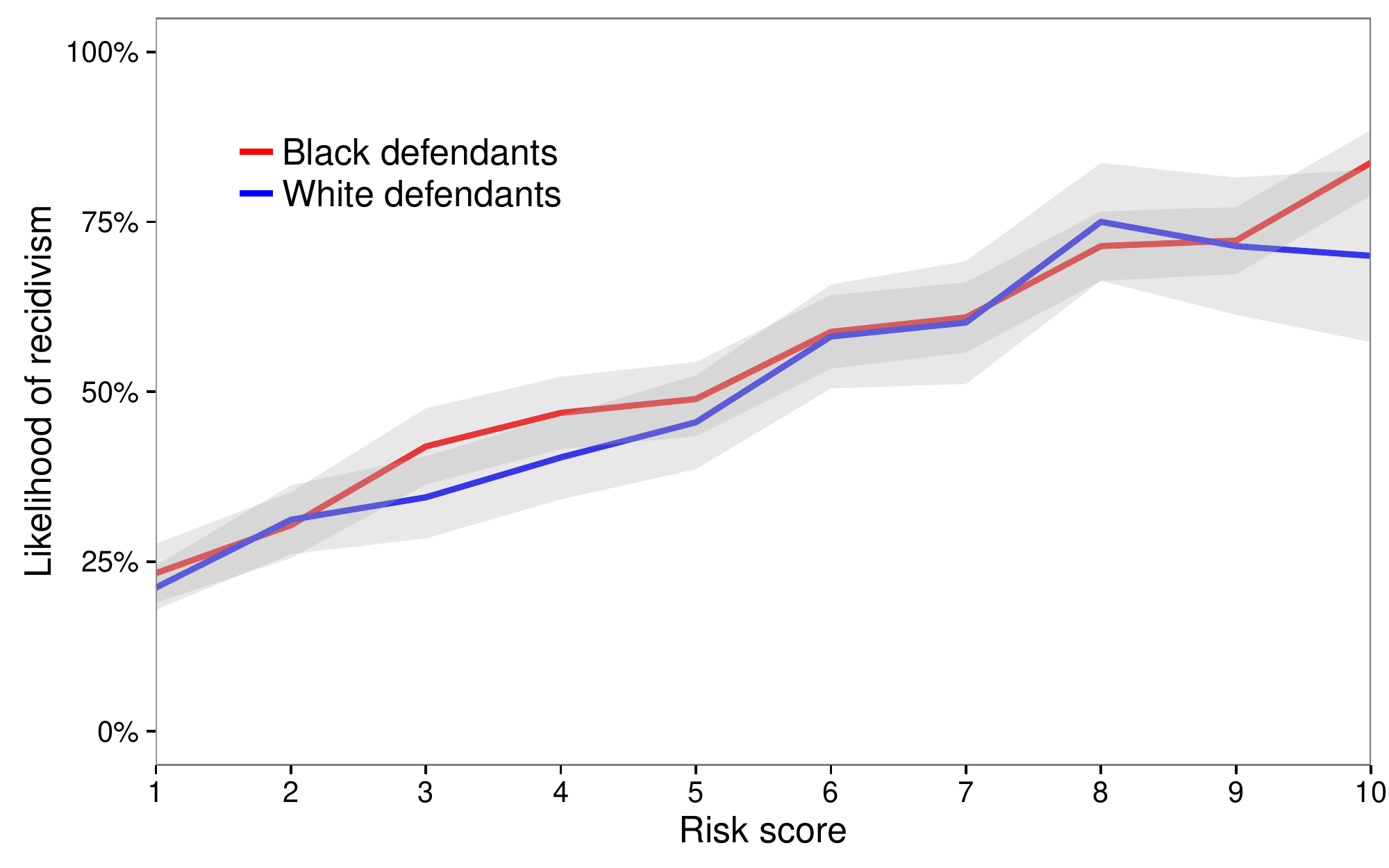}
\vspace{2mm}
\caption{Recidivism rate by COMPAS risk score and race. White and black defendants with the same risk score are roughly equally likely to reoffend, indicating that the scores are calibrated. 
The $y$-axis shows the proportion of defendants re-arrested for any crime, including non-violent offenses; 
the gray bands show 95\% confidence intervals.
}
\label{fig:calibration}
\end{figure}

However, given only scores $s(x)$ and outcomes $y$, it is impossible to determine whether the scores are accurate estimates of $p_{Y|X}$ or have been strategically designed to produce racial disparities. \citet{hardt2016} make a similar observation in their discussion of ``oblivious'' measures.
Consider a hypothetical situation where a malicious decision maker wants to release all white defendants, even if they are high risk. 
To shield himself from claims of discrimination, he applies a facially neutral 30\% threshold to defendants regardless of race. 
Suppose that 20\% of blacks recidivate, and the decision-maker's algorithm uses additional information, such as prior arrests, to partition blacks into three risk categories: low risk (10\% chance of reoffending), average risk (20\% chance), and high risk (40\% chance). 
Further suppose that whites are just as risky as blacks overall (20\% of them reoffend), but the decision maker ignores individual characteristics and labels every white defendant average risk. This algorithm is calibrated, as both whites and blacks labeled average risk reoffend 20\% of the time. 
However, all white defendants fall below the decision threshold, so none are detained. 
By systematically ignoring information that could be used to distinguish between white defendants, 
the decision maker has succeeded in discriminating while using a single threshold applied to calibrated scores.

Figure~\ref{fig:calibration_counterexample} illustrates a general method for constructing such discriminatory scores
from true risk estimates.
We start by adding noise to the true scores (black curve) of the group that we wish to treat favorably---in the figure we use $\text{N}(0, 0.5)$ noise.
We then use the perturbed scores to predict the outcomes $y_i$ via a logistic regression model.
The resulting model predictions (red curve) are more tightly clustered around their mean, 
since adding noise removes information. 
Consequently, under the transformed scores, no one in the group lies above the decision threshold, indicated by the vertical line.
The key point is that the red curve is a perfectly plausible distribution of risk: without further information, 
one cannot determine whether the risk model was fit on 
input data that were truly noisy, 
or whether noise was added to the inputs to produce disparities.

These examples relate to the historical practice of \emph{redlining}, in which lending decisions were intentionally based only on 
coarse information---usually neighborhood---in order to deny loans to well-qualified minorities~\cite{berkovec1994race}.
Since even creditworthy minorities often resided 
in neighborhoods with low average income,
lenders could deny their applications by adhering to a facially neutral policy of not serving low-income areas.
In the case of redlining, one discriminates by
ignoring information about the disfavored group;
in the pretrial setting, one ignores information about the favored group. Both strategies, however, operate under the same general principle.

There is no evidence to suggest that organizations have intentionally ignored relevant information when constructing risk scores.
Similar effects, however, may also arise through negligence or unintentional oversights.
Indeed, we found in Section~\ref{sec:empirical} that we could improve the predictive power of the Broward County COMPAS scores
with a standard statistical model.
To ensure an algorithm is equitable, it is thus
important to inspect the algorithm itself and not just the decisions it produces.

\begin{figure}[t]
\centering
\includegraphics[width=\columnwidth, trim=0 0.7cm 0 0]{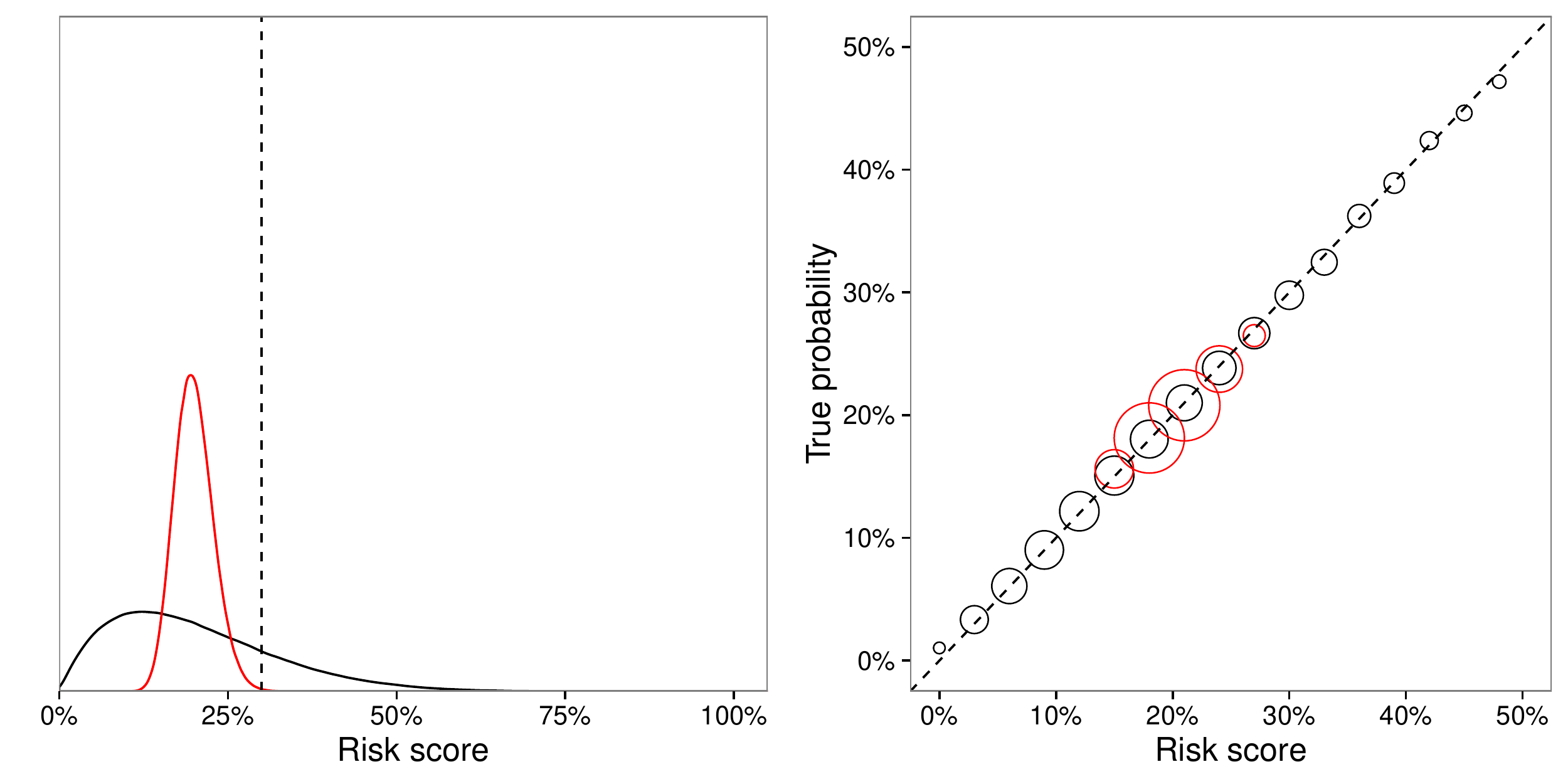}
\caption{Calibration is insufficient to assess discrimination. 
In the left plot, the black line shows the distribution of risk in a hypothetical population,
and the red line shows strategically altered 
risk estimates in the same population. 
Both sets of risk scores are calibrated (right plot), 
but the altered risk scores are less informative and as a result guarantee that no defendants fall above the detention threshold (dashed vertical line).}
\label{fig:calibration_counterexample}
\end{figure}

\section{Discussion}
\label{sec:discussion}
Maximizing public safety requires detaining all individuals deemed sufficiently likely to commit a violent crime, regardless of race.\pagebreak[0]
However, to satisfy common metrics of fairness, one must set multiple, race-specific thresholds.\pagebreak[0]
There is thus an inherent tension between minimizing expected violent crime and satisfying common notions of fairness.\pagebreak[0]
This tension is real:\pagebreak[0]
by analyzing data from Broward County, 
we find that optimizing for public safety yields stark racial disparities;
conversely, satisfying past fairness definitions means 
releasing more high-risk defendants, 
adversely affecting public safety.

Policymakers face a difficult and consequential choice,
and it is ultimately unclear what course of action is best in any given situation.
We note, however, one important consideration:
with race-specific thresholds, a black defendant may be released while an equally risky white defendant is detained. 
Such racial classifications would likely trigger \emph{strict scrutiny}~\cite{fisher2016}, the most stringent standard of judicial review used by U.S. courts under the Equal Protection Clause of the Fourteenth Amendment.
A single-threshold rule thus maximizes public safety while satisfying a core constitutional law rule, bolstering the case in its favor.

To some extent, concerns embodied by past fairness definitions can be addressed while still adopting a single-threshold rule. For example, by collecting more data and accordingly increasing the accuracy of risk estimates, one can lower error rates. 
Further, one could raise the threshold for detaining defendants, reducing the number of people erroneously detained from all race groups. Finally, one could change the decision such that classification errors are less costly. For example, rather than being held in jail, risky defendants might be required to participate in community supervision programs.

When evaluating policy options, it is important to consider how well risk scores capture the salient costs and benefits of the decision. For example, though we might want to minimize violent crime conducted by defendants awaiting trial, we typically only observe crime that results in an arrest. But arrests are an imperfect proxy. Heavier policing in minority neighborhoods might lead to black defendants being arrested more often than whites who commit the same crime~\cite{lum2016predict}. Poor outcome data might thus cause one to systematically underestimate the risk posed by white defendants.
This concern is mitigated when the outcome $y$ is serious crime---rather than minor offenses---since such incidents are less susceptible to biased observation. In particular, \citet{skeem2015risk} note that the racial distribution of individuals arrested for violent offenses is in line with the racial distribution of offenders inferred from victim reports and also in line with self-reported offending data. 

One might similarly worry that the features $x$ are biased in the sense that factors are not equally predictive across race groups, a phenomenon known as \emph{subgroup validity}~\cite{ayres2002outcome}. For example, housing stability might be less predictive of recidivism for minorities than for whites. If the vector of features $x$ includes race, an individual's risk score $p_{y|x}$ is in theory statistically robust to this issue, and for this reason some have argued race should be included in risk models~\cite{berk_2009}. However, explicitly including race as an input feature raises legal and policy complications, and as such it is common to simply exclude features with differential predictive power~\cite{danner2015}. While perhaps a reasonable strategy in practice, we note that discarding information may inadvertently lead to the redlining effects we discuss in Section~\ref{sec:calibration}.

Risk scores might also fail to accurately capture costs in specific, idiosyncratic cases. Detaining a defendant who is the sole caretaker of her children arguably incurs higher social costs than detaining a defendant without children. Discretionary consideration of individual cases might thus be justified, provided that such discretion does not also introduce bias.
Further, the immediate utility of a decision rule might be a poor measure of its long-term costs and benefits.
For example, in the context of credit extensions, offering loans preferentially to minorities might ultimately lead to a more productive distribution of wealth, combating harms from historical under-investment in minority communities. 

Finally, we note that some decisions are better thought of as group rather than individual choices, limiting the applicability of the framework we have been considering. For example, when universities admit students, they often aim to select the best group, not simply the best individual candidates, and may thus decide to deviate from a single-threshold rule in order to create diverse communities with varied perspectives and backgrounds \cite{page2008difference}.

Experts increasingly rely on algorithmic decision aids in diverse settings, including law enforcement, education, employment, and medicine~\cite{barocas2016big,berk2012criminal,goel2016-risky,goel2016precinct,jung2017}.
Algorithms have the potential to improve the efficiency and equity of decisions,
but their design and application 
raise complex questions for researchers and policymakers.
By clarifying the implications of competing notions of algorithmic fairness, we hope our analysis 
fosters discussion and informs policy.

\begin{acks}
We thank Imanol Arrieta Ibarra, Ang{\`e}le Christin, 
Alexander Gelber, Andrew Gelman, 
Jan Overgoor, Ravi Shroff, and Jennifer Skeem 
for their helpful comments.
This work was supported in part by the John S. and James L. Knight Foundation,
and by the Hellman Fellows Fund.
Data and code to reproduce our results are available at \url{https://github.com/5harad/cost-of-fairness}.
\end{acks}

\balance
\bibliographystyle{ACM-Reference-Format}
\bibliography{algobias}


\begin{thebibliography}{00}


\ifx \showCODEN    \undefined \def \showCODEN     #1{\unskip}     \fi
\ifx \showDOI      \undefined \def \showDOI       #1{{\tt DOI:}\penalty0{#1}\ }
  \fi
\ifx \showISBNx    \undefined \def \showISBNx     #1{\unskip}     \fi
\ifx \showISBNxiii \undefined \def \showISBNxiii  #1{\unskip}     \fi
\ifx \showISSN     \undefined \def \showISSN      #1{\unskip}     \fi
\ifx \showLCCN     \undefined \def \showLCCN      #1{\unskip}     \fi
\ifx \shownote     \undefined \def \shownote      #1{#1}          \fi
\ifx \showarticletitle \undefined \def \showarticletitle #1{#1}   \fi
\ifx \showURL      \undefined \def \showURL       #1{#1}          \fi
\providecommand\bibfield[2]{#2}
\providecommand\bibinfo[2]{#2}
\providecommand\natexlab[1]{#1}
\providecommand\showeprint[2][]{arXiv:#2}

\bibitem[\protect\citeauthoryear{Ahlman and Kurtz}{Ahlman and Kurtz}{2009}]%
        {ahlman2009}
\bibfield{author}{\bibinfo{person}{Lindsay~C Ahlman} {and}
  \bibinfo{person}{Ellen~M Kurtz}.} \bibinfo{year}{2009}\natexlab{}.
\newblock \showarticletitle{The APPD randomized controlled trial in low risk
  supervision: The effects on low risk supervision on rearrest}.
\newblock \bibinfo{journal}{{\em Philadelphia Adult Probation and Parole
  Department\/}} (\bibinfo{year}{2009}).
\newblock


\bibitem[\protect\citeauthoryear{Angwin, Larson, Mattu, and Kirchner}{Angwin
  et~al\mbox{.}}{2016}]%
        {angwin2016}
\bibfield{author}{\bibinfo{person}{Julia Angwin}, \bibinfo{person}{Jeff
  Larson}, \bibinfo{person}{Surya Mattu}, {and} \bibinfo{person}{Lauren
  Kirchner}.} \bibinfo{year}{2016}\natexlab{}.
\newblock \showarticletitle{Machine bias: There's software used across the
  country to predict future criminals. and it's biased against blacks}.
\newblock \bibinfo{journal}{{\em ProPublica\/}} (\bibinfo{date}{5}
  \bibinfo{year}{2016}).
\newblock


\bibitem[\protect\citeauthoryear{Anwar and Fang}{Anwar and Fang}{2006}]%
        {anwar2006}
\bibfield{author}{\bibinfo{person}{Shamena Anwar} {and}
  \bibinfo{person}{Hanming Fang}.} \bibinfo{year}{2006}\natexlab{}.
\newblock \showarticletitle{An Alternative Test of Racial Prejudice in Motor
  Vehicle Searches: Theory and Evidence}.
\newblock \bibinfo{journal}{{\em The American Economic Review\/}}
  (\bibinfo{year}{2006}).
\newblock


\bibitem[\protect\citeauthoryear{Ayres}{Ayres}{2002}]%
        {ayres2002outcome}
\bibfield{author}{\bibinfo{person}{Ian Ayres}.}
  \bibinfo{year}{2002}\natexlab{}.
\newblock \showarticletitle{Outcome tests of racial disparities in police
  practices}.
\newblock \bibinfo{journal}{{\em Justice research and Policy\/}}
  \bibinfo{volume}{4}, \bibinfo{number}{1-2} (\bibinfo{year}{2002}),
  \bibinfo{pages}{131--142}.
\newblock


\bibitem[\protect\citeauthoryear{Barocas and Selbst}{Barocas and
  Selbst}{2016}]%
        {barocas2016big}
\bibfield{author}{\bibinfo{person}{Solon Barocas} {and}
  \bibinfo{person}{Andrew~D Selbst}.} \bibinfo{year}{2016}\natexlab{}.
\newblock \showarticletitle{Big data's disparate impact}.
\newblock \bibinfo{journal}{{\em California Law Review\/}}
  \bibinfo{volume}{104} (\bibinfo{year}{2016}).
\newblock


\bibitem[\protect\citeauthoryear{Becker}{Becker}{1957}]%
        {becker1957}
\bibfield{author}{\bibinfo{person}{Gary~S Becker}.}
  \bibinfo{year}{1957}\natexlab{}.
\newblock \bibinfo{booktitle}{{\em The economics of discrimination}}.
\newblock \bibinfo{publisher}{University of Chicago Press}.
\newblock


\bibitem[\protect\citeauthoryear{Berk}{Berk}{2009}]%
        {berk_2009}
\bibfield{author}{\bibinfo{person}{Richard Berk}.}
  \bibinfo{year}{2009}\natexlab{}.
\newblock \showarticletitle{The role of race in forecasts of violent crime}.
\newblock \bibinfo{journal}{{\em Race and Social Problems\/}}
  (\bibinfo{year}{2009}), \bibinfo{pages}{131--242}.
\newblock


\bibitem[\protect\citeauthoryear{Berk}{Berk}{2012}]%
        {berk2012criminal}
\bibfield{author}{\bibinfo{person}{Richard Berk}.}
  \bibinfo{year}{2012}\natexlab{}.
\newblock \bibinfo{booktitle}{{\em Criminal justice forecasts of risk: a
  machine learning approach}}.
\newblock \bibinfo{publisher}{Springer Science \& Business Media}.
\newblock


\bibitem[\protect\citeauthoryear{Berk, Heidari, Jabbari, Kearns, and Roth}{Berk
  et~al\mbox{.}}{2017}]%
        {berk2017}
\bibfield{author}{\bibinfo{person}{Richard Berk}, \bibinfo{person}{Hoda
  Heidari}, \bibinfo{person}{Shahin Jabbari}, \bibinfo{person}{Michael Kearns},
  {and} \bibinfo{person}{Aaron Roth}.} \bibinfo{year}{2017}\natexlab{}.
\newblock \showarticletitle{Fairness in Criminal Justice Risk Assessments: The
  State of the Art}.
\newblock \bibinfo{journal}{{\em working paper\/}} (\bibinfo{year}{2017}).
\newblock


\bibitem[\protect\citeauthoryear{Berk, Sherman, Barnes, Kurtz, and Ahlman}{Berk
  et~al\mbox{.}}{2009}]%
        {berk2009}
\bibfield{author}{\bibinfo{person}{Richard Berk}, \bibinfo{person}{Lawrence
  Sherman}, \bibinfo{person}{Geoffrey Barnes}, \bibinfo{person}{Ellen Kurtz},
  {and} \bibinfo{person}{Lindsay Ahlman}.} \bibinfo{year}{2009}\natexlab{}.
\newblock \showarticletitle{Forecasting murder within a population of
  probationers and parolees: a high stakes application of statistical
  learning}.
\newblock \bibinfo{journal}{{\em Journal of the Royal Statistical Society:
  Series A (Statistics in Society)\/}} \bibinfo{volume}{172},
  \bibinfo{number}{1} (\bibinfo{year}{2009}), \bibinfo{pages}{191--211}.
\newblock


\bibitem[\protect\citeauthoryear{Berkovec, Canner, Gabriel, and
  Hannan}{Berkovec et~al\mbox{.}}{1994}]%
        {berkovec1994race}
\bibfield{author}{\bibinfo{person}{James~A Berkovec}, \bibinfo{person}{Glenn~B
  Canner}, \bibinfo{person}{Stuart~A Gabriel}, {and} \bibinfo{person}{Timothy~H
  Hannan}.} \bibinfo{year}{1994}\natexlab{}.
\newblock \showarticletitle{Race, redlining, and residential mortgage loan
  performance}.
\newblock \bibinfo{journal}{{\em The Journal of Real Estate Finance and
  Economics\/}} \bibinfo{volume}{9}, \bibinfo{number}{3}
  (\bibinfo{year}{1994}), \bibinfo{pages}{263--294}.
\newblock


\bibitem[\protect\citeauthoryear{Chouldechova}{Chouldechova}{2016}]%
        {chouldechova2016fair}
\bibfield{author}{\bibinfo{person}{Alexandra Chouldechova}.}
  \bibinfo{year}{2016}\natexlab{}.
\newblock \showarticletitle{Fair prediction with disparate impact: A study of
  bias in recidivism prediction instruments}.
\newblock \bibinfo{journal}{{\em arXiv preprint arXiv:1610.07524\/}}
  (\bibinfo{year}{2016}).
\newblock


\bibitem[\protect\citeauthoryear{Christin, Rosenblat, and boyd}{Christin
  et~al\mbox{.}}{2015}]%
        {christin2015}
\bibfield{author}{\bibinfo{person}{Ang{\`e}le Christin}, \bibinfo{person}{Alex
  Rosenblat}, {and} \bibinfo{person}{danah boyd}.}
  \bibinfo{year}{2015}\natexlab{}.
\newblock \showarticletitle{Courts and Predictive Algorithms}.
\newblock \bibinfo{journal}{{\em Data \& Civil Rights: Criminal Justice and
  Civil Rights Primer\/}} (\bibinfo{year}{2015}).
\newblock


\bibitem[\protect\citeauthoryear{Danner, VanNostrand, and Spruance}{Danner
  et~al\mbox{.}}{2015}]%
        {danner2015}
\bibfield{author}{\bibinfo{person}{Mona Danner}, \bibinfo{person}{Marie
  VanNostrand}, {and} \bibinfo{person}{Lisa Spruance}.} \bibinfo{year}{August,
  2015}\natexlab{}.
\newblock \showarticletitle{Risk-{B}ased {P}retrial {R}elease {R}ecommendation
  and {S}upervision {G}uidelines}.
\newblock  (\bibinfo{year}{August, 2015}).
\newblock


\bibitem[\protect\citeauthoryear{Dwork, Hardt, Pitassi, Reingold, and
  Zemel}{Dwork et~al\mbox{.}}{2012}]%
        {dwork2012}
\bibfield{author}{\bibinfo{person}{Cynthia Dwork}, \bibinfo{person}{Moritz
  Hardt}, \bibinfo{person}{Toniann Pitassi}, \bibinfo{person}{Omer Reingold},
  {and} \bibinfo{person}{Richard Zemel}.} \bibinfo{year}{2012}\natexlab{}.
\newblock \showarticletitle{Fairness through awareness}. In
  \bibinfo{booktitle}{{\em Proceedings of the 3rd Innovations in Theoretical
  Computer Science Conference}}. ACM, \bibinfo{pages}{214--226}.
\newblock


\bibitem[\protect\citeauthoryear{Feldman, Friedler, Moeller, Scheidegger, and
  Venkatasubramanian}{Feldman et~al\mbox{.}}{2015}]%
        {feldman2015}
\bibfield{author}{\bibinfo{person}{Michael Feldman}, \bibinfo{person}{Sorelle~A
  Friedler}, \bibinfo{person}{John Moeller}, \bibinfo{person}{Carlos
  Scheidegger}, {and} \bibinfo{person}{Suresh Venkatasubramanian}.}
  \bibinfo{year}{2015}\natexlab{}.
\newblock \showarticletitle{Certifying and removing disparate impact}. In
  \bibinfo{booktitle}{{\em Proceedings of the 21th ACM SIGKDD International
  Conference on Knowledge Discovery and Data Mining}}. ACM,
  \bibinfo{pages}{259--268}.
\newblock


\bibitem[\protect\citeauthoryear{Fish, Kun, and Lelkes}{Fish
  et~al\mbox{.}}{2016}]%
        {fish16}
\bibfield{author}{\bibinfo{person}{Benjamin Fish}, \bibinfo{person}{Jeremy
  Kun}, {and} \bibinfo{person}{{\'{A}}d{\'{a}}m~D{\'{a}}niel Lelkes}.}
  \bibinfo{year}{2016}\natexlab{}.
\newblock \showarticletitle{A Confidence-Based Approach for Balancing Fairness
  and Accuracy}.
\newblock \bibinfo{journal}{{\em CoRR\/}}  \bibinfo{volume}{abs/1601.05764}
  (\bibinfo{year}{2016}).
\newblock


\bibitem[\protect\citeauthoryear{{Fisher v. University of Texas at
  Austin}}{{Fisher v. University of Texas at Austin}}{2016}]%
        {fisher2016}
\bibfield{author}{\bibinfo{person}{{Fisher v. University of Texas at Austin}}.}
  \bibinfo{year}{2016}\natexlab{}.
\newblock   \bibinfo{volume}{136 S. Ct. 2198, 2208} (\bibinfo{year}{2016}).
\newblock


\bibitem[\protect\citeauthoryear{Goel, Rao, and Shroff}{Goel
  et~al\mbox{.}}{2016a}]%
        {goel2016-risky}
\bibfield{author}{\bibinfo{person}{Sharad Goel}, \bibinfo{person}{Justin Rao},
  {and} \bibinfo{person}{Ravi Shroff}.} \bibinfo{year}{2016}\natexlab{a}.
\newblock \showarticletitle{Personalized Risk Assessments in the Criminal
  Justice System}.
\newblock \bibinfo{journal}{{\em The American Economic Review\/}}
  \bibinfo{volume}{106}, \bibinfo{number}{5} (\bibinfo{year}{2016}),
  \bibinfo{pages}{119--123}.
\newblock


\bibitem[\protect\citeauthoryear{Goel, Rao, and Shroff}{Goel
  et~al\mbox{.}}{2016b}]%
        {goel2016precinct}
\bibfield{author}{\bibinfo{person}{Sharad Goel}, \bibinfo{person}{Justin~M
  Rao}, {and} \bibinfo{person}{Ravi Shroff}.} \bibinfo{year}{2016}\natexlab{b}.
\newblock \showarticletitle{Precinct or Prejudice? {U}nderstanding Racial
  Disparities in {New York City's} Stop-and-Frisk Policy}.
\newblock \bibinfo{journal}{{\em Annals of Applied Statistics\/}}
  \bibinfo{volume}{10}, \bibinfo{number}{1} (\bibinfo{year}{2016}),
  \bibinfo{pages}{365--394}.
\newblock


\bibitem[\protect\citeauthoryear{Hajian, Domingo-Ferrer, and
  Martinez-Balleste}{Hajian et~al\mbox{.}}{2011}]%
        {hajian2011discrimination}
\bibfield{author}{\bibinfo{person}{Sara Hajian}, \bibinfo{person}{Josep
  Domingo-Ferrer}, {and} \bibinfo{person}{Antoni Martinez-Balleste}.}
  \bibinfo{year}{2011}\natexlab{}.
\newblock \showarticletitle{Discrimination prevention in data mining for
  intrusion and crime detection}. In \bibinfo{booktitle}{{\em Computational
  Intelligence in Cyber Security (CICS), 2011 IEEE Symposium on}}. IEEE,
  \bibinfo{pages}{47--54}.
\newblock


\bibitem[\protect\citeauthoryear{Hardt, Price, and Srebro}{Hardt
  et~al\mbox{.}}{2016}]%
        {hardt2016}
\bibfield{author}{\bibinfo{person}{Moritz Hardt}, \bibinfo{person}{Eric Price},
  {and} \bibinfo{person}{Nati Srebro}.} \bibinfo{year}{2016}\natexlab{}.
\newblock \showarticletitle{Equality of opportunity in supervised learning}. In
  \bibinfo{booktitle}{{\em Advances In Neural Information Processing Systems}}.
  \bibinfo{pages}{3315--3323}.
\newblock


\bibitem[\protect\citeauthoryear{Holder}{Holder}{2014}]%
        {holder2014}
\bibfield{author}{\bibinfo{person}{Eric Holder}.}
  \bibinfo{year}{2014}\natexlab{}.
\newblock   (\bibinfo{year}{2014}).
\newblock
\newblock
\shownote{Remarks at the National Association of Criminal Defense Lawyers 57th
  Annual Meeting.}


\bibitem[\protect\citeauthoryear{Jabbari, Joseph, Kearns, Morgenstern, and
  Roth}{Jabbari et~al\mbox{.}}{2016}]%
        {jabbari2016fair}
\bibfield{author}{\bibinfo{person}{Shahin Jabbari}, \bibinfo{person}{Matthew
  Joseph}, \bibinfo{person}{Michael Kearns}, \bibinfo{person}{Jamie
  Morgenstern}, {and} \bibinfo{person}{Aaron Roth}.}
  \bibinfo{year}{2016}\natexlab{}.
\newblock \showarticletitle{Fair Learning in Markovian Environments}.
\newblock \bibinfo{journal}{{\em arXiv preprint arXiv:1611.03071\/}}
  (\bibinfo{year}{2016}).
\newblock


\bibitem[\protect\citeauthoryear{Joseph, Kearns, Morgenstern, and Roth}{Joseph
  et~al\mbox{.}}{2016}]%
        {joseph2016fairness}
\bibfield{author}{\bibinfo{person}{Matthew Joseph}, \bibinfo{person}{Michael
  Kearns}, \bibinfo{person}{Jamie~H Morgenstern}, {and} \bibinfo{person}{Aaron
  Roth}.} \bibinfo{year}{2016}\natexlab{}.
\newblock \showarticletitle{Fairness in Learning: Classic and Contextual
  Bandits}.
\newblock In \bibinfo{booktitle}{{\em Advances in Neural Information Processing
  Systems 29}}. \bibinfo{publisher}{Curran Associates, Inc.},
  \bibinfo{pages}{325--333}.
\newblock


\bibitem[\protect\citeauthoryear{Jung, Concannon, Shroff, Goel, and
  Goldstein}{Jung et~al\mbox{.}}{2017}]%
        {jung2017}
\bibfield{author}{\bibinfo{person}{Jongbin Jung}, \bibinfo{person}{Connor
  Concannon}, \bibinfo{person}{Ravi Shroff}, \bibinfo{person}{Sharad Goel},
  {and} \bibinfo{person}{Daniel~G. Goldstein}.}
  \bibinfo{year}{2017}\natexlab{}.
\newblock \bibinfo{title}{Simple rules for complex decisions}.
  (\bibinfo{year}{2017}).
\newblock
\newblock
\shownote{Working paper.}


\bibitem[\protect\citeauthoryear{Kamiran, {\v{Z}}liobait{\.e}, and
  Calders}{Kamiran et~al\mbox{.}}{2013}]%
        {kamiran2013}
\bibfield{author}{\bibinfo{person}{Faisal Kamiran}, \bibinfo{person}{Indr{\.e}
  {\v{Z}}liobait{\.e}}, {and} \bibinfo{person}{Toon Calders}.}
  \bibinfo{year}{2013}\natexlab{}.
\newblock \showarticletitle{Quantifying explainable discrimination and removing
  illegal discrimination in automated decision making}.
\newblock \bibinfo{journal}{{\em Knowledge and information systems\/}}
  \bibinfo{volume}{35}, \bibinfo{number}{3} (\bibinfo{year}{2013}),
  \bibinfo{pages}{613--644}.
\newblock


\bibitem[\protect\citeauthoryear{Kamishima, Akaho, Asoh, and Sakuma}{Kamishima
  et~al\mbox{.}}{2012}]%
        {kamishima2012}
\bibfield{author}{\bibinfo{person}{Toshihiro Kamishima},
  \bibinfo{person}{Shotaro Akaho}, \bibinfo{person}{Hideki Asoh}, {and}
  \bibinfo{person}{Jun Sakuma}.} \bibinfo{year}{2012}\natexlab{}.
\newblock \showarticletitle{Fairness-aware classifier with prejudice remover
  regularizer}. In \bibinfo{booktitle}{{\em Joint European Conference on
  Machine Learning and Knowledge Discovery in Databases}}. Springer,
  \bibinfo{pages}{35--50}.
\newblock


\bibitem[\protect\citeauthoryear{Kleinberg, Mullainathan, and
  Raghavan}{Kleinberg et~al\mbox{.}}{2017}]%
        {kleinberg2016inherent}
\bibfield{author}{\bibinfo{person}{Jon Kleinberg}, \bibinfo{person}{Sendhil
  Mullainathan}, {and} \bibinfo{person}{Manish Raghavan}.}
  \bibinfo{year}{2017}\natexlab{}.
\newblock \showarticletitle{Inherent trade-offs in the fair determination of
  risk scores}.
\newblock \bibinfo{journal}{{\em Proceedings of Innovations in Theoretical
  Computer Science (ITCS)\/}} (\bibinfo{year}{2017}).
\newblock


\bibitem[\protect\citeauthoryear{Larson, Mattu, Kirchner, and Angwin}{Larson
  et~al\mbox{.}}{2016}]%
        {larson2016}
\bibfield{author}{\bibinfo{person}{Jeff Larson}, \bibinfo{person}{Surya Mattu},
  \bibinfo{person}{Lauren Kirchner}, {and} \bibinfo{person}{Julia Angwin}.}
  \bibinfo{year}{2016}\natexlab{}.
\newblock \showarticletitle{How We Analyzed the COMPAS Recidivism Algorithm}.
\newblock \bibinfo{journal}{{\em ProPublica\/}} (\bibinfo{date}{5}
  \bibinfo{year}{2016}).
\newblock


\bibitem[\protect\citeauthoryear{Lum and Isaac}{Lum and Isaac}{2016}]%
        {lum2016predict}
\bibfield{author}{\bibinfo{person}{Kristian Lum} {and} \bibinfo{person}{William
  Isaac}.} \bibinfo{year}{2016}\natexlab{}.
\newblock \showarticletitle{To predict and serve?}
\newblock \bibinfo{journal}{{\em Significance\/}} \bibinfo{volume}{13},
  \bibinfo{number}{5} (\bibinfo{year}{2016}), \bibinfo{pages}{14--19}.
\newblock


\bibitem[\protect\citeauthoryear{Monahan and Skeem}{Monahan and Skeem}{2016}]%
        {monahan2016}
\bibfield{author}{\bibinfo{person}{J Monahan} {and} \bibinfo{person}{JL
  Skeem}.} \bibinfo{year}{2016}\natexlab{}.
\newblock \showarticletitle{Risk Assessment in Criminal Sentencing.}
\newblock \bibinfo{journal}{{\em Annual review of clinical psychology\/}}
  \bibinfo{volume}{12} (\bibinfo{year}{2016}), \bibinfo{pages}{489}.
\newblock


\bibitem[\protect\citeauthoryear{Page}{Page}{2008}]%
        {page2008difference}
\bibfield{author}{\bibinfo{person}{Scott~E Page}.}
  \bibinfo{year}{2008}\natexlab{}.
\newblock \bibinfo{booktitle}{{\em The difference: How the power of diversity
  creates better groups, firms, schools, and societies}}.
\newblock \bibinfo{publisher}{Princeton University Press}.
\newblock


\bibitem[\protect\citeauthoryear{Pierson, Corbett-Davies, and Goel}{Pierson
  et~al\mbox{.}}{2017}]%
        {pierson2017}
\bibfield{author}{\bibinfo{person}{Emma Pierson}, \bibinfo{person}{Sam
  Corbett-Davies}, {and} \bibinfo{person}{Sharad Goel}.}
  \bibinfo{year}{2017}\natexlab{}.
\newblock \bibinfo{title}{Fast threshold tests for detecting discrimination}.
  (\bibinfo{year}{2017}).
\newblock
\newblock
\shownote{Working paper available at \url{https://arxiv.org/abs/1702.08536}.}


\bibitem[\protect\citeauthoryear{Platt}{Platt}{1999}]%
        {platt99}
\bibfield{author}{\bibinfo{person}{John~C. Platt}.}
  \bibinfo{year}{1999}\natexlab{}.
\newblock \showarticletitle{Probabilistic Outputs for Support Vector Machines
  and Comparisons to Regularized Likelihood Methods}. In
  \bibinfo{booktitle}{{\em Advances in Large Margin Classifiers}}.
  \bibinfo{pages}{61--74}.
\newblock


\bibitem[\protect\citeauthoryear{Romei and Ruggieri}{Romei and
  Ruggieri}{2014}]%
        {romei2014multidisciplinary}
\bibfield{author}{\bibinfo{person}{Andrea Romei} {and}
  \bibinfo{person}{Salvatore Ruggieri}.} \bibinfo{year}{2014}\natexlab{}.
\newblock \showarticletitle{A multidisciplinary survey on discrimination
  analysis}.
\newblock \bibinfo{journal}{{\em The Knowledge Engineering Review\/}}
  \bibinfo{volume}{29}, \bibinfo{number}{05} (\bibinfo{year}{2014}),
  \bibinfo{pages}{582--638}.
\newblock


\bibitem[\protect\citeauthoryear{Simoiu, Corbett-Davies, and Goel}{Simoiu
  et~al\mbox{.}}{2017}]%
        {simoiu2017}
\bibfield{author}{\bibinfo{person}{Camelia Simoiu}, \bibinfo{person}{Sam
  Corbett-Davies}, {and} \bibinfo{person}{Sharad Goel}.}
  \bibinfo{year}{2017}\natexlab{}.
\newblock \showarticletitle{The problem of infra-marginality in outcome tests
  for discrimination}.
\newblock \bibinfo{journal}{{\em Annals of Applied Statistics\/}}
  (\bibinfo{year}{2017}).
\newblock
\newblock
\shownote{Forthcoming.}


\bibitem[\protect\citeauthoryear{Skeem and Lowencamp}{Skeem and
  Lowencamp}{2016}]%
        {skeem2015risk}
\bibfield{author}{\bibinfo{person}{Jennifer~L. Skeem} {and}
  \bibinfo{person}{Christopher~T. Lowencamp}.} \bibinfo{year}{2016}\natexlab{}.
\newblock \showarticletitle{Risk, Race, and Recidivism: Predictive Bias and
  Disparate Impact}.
\newblock \bibinfo{journal}{{\em Criminology\/}} \bibinfo{volume}{54},
  \bibinfo{number}{4} (\bibinfo{year}{2016}), \bibinfo{pages}{680--712}.
\newblock


\bibitem[\protect\citeauthoryear{Starr}{Starr}{2014}]%
        {starr2014evidence}
\bibfield{author}{\bibinfo{person}{Sonja~B Starr}.}
  \bibinfo{year}{2014}\natexlab{}.
\newblock \showarticletitle{Evidence-based sentencing and the scientific
  rationalization of discrimination}.
\newblock \bibinfo{journal}{{\em Stan. L. Rev.\/}}  \bibinfo{volume}{66}
  (\bibinfo{year}{2014}), \bibinfo{pages}{803}.
\newblock


\bibitem[\protect\citeauthoryear{Zafar, Valera, Rodriguez, and Gummadi}{Zafar
  et~al\mbox{.}}{2017}]%
        {zafar2016fairness}
\bibfield{author}{\bibinfo{person}{Muhammad~Bilal Zafar},
  \bibinfo{person}{Isabel Valera}, \bibinfo{person}{Manuel~Gomez Rodriguez},
  {and} \bibinfo{person}{Krishna~P Gummadi}.} \bibinfo{year}{2017}\natexlab{}.
\newblock \showarticletitle{Fairness Beyond Disparate Treatment \& Disparate
  Impact: Learning Classification without Disparate Mistreatment}. In
  \bibinfo{booktitle}{{\em Proceedings of the 26th International World Wide Web
  Conference}}.
\newblock


\bibitem[\protect\citeauthoryear{Zemel, Wu, Swersky, Pitassi, and Dwork}{Zemel
  et~al\mbox{.}}{2013}]%
        {zemel2013}
\bibfield{author}{\bibinfo{person}{Rich Zemel}, \bibinfo{person}{Yu Wu},
  \bibinfo{person}{Kevin Swersky}, \bibinfo{person}{Toni Pitassi}, {and}
  \bibinfo{person}{Cynthia Dwork}.} \bibinfo{year}{2013}\natexlab{}.
\newblock \showarticletitle{Learning Fair Representations}. In
  \bibinfo{booktitle}{{\em Proceedings of The 30th International Conference on
  Machine Learning}}. \bibinfo{pages}{325--333}.
\newblock


\bibitem[\protect\citeauthoryear{Zliobaite}{Zliobaite}{2017}]%
        {vzliobaitemeasuring}
\bibfield{author}{\bibinfo{person}{Indre Zliobaite}.}
  \bibinfo{year}{2017}\natexlab{}.
\newblock \showarticletitle{Measuring discrimination in algorithmic decision
  making}.
\newblock \bibinfo{journal}{{\em Data Mining and Knowledge Discovery\/}}
  (\bibinfo{year}{2017}).
\newblock


\end{thebibliography}

\end{document}